\documentclass[12pt]{amsart}

\usepackage[centertags]{amsmath}
\usepackage{amsthm,amsfonts}

\usepackage{amssymb}
\usepackage{epsfig}
\usepackage{amssymb,latexsym}
\usepackage{color}

\begin{document}

\theoremstyle{plain}
\newtheorem{theorem}{Theorem}[section]
\newtheorem{lemma}[theorem]{Lemma}
\newtheorem{corollary}[theorem]{Corollary}
\newtheorem{proposition}[theorem]{Proposition}
\newtheorem{question}[theorem]{Question}
\theoremstyle{definition}
\newtheorem{notations}[theorem]{Notations}
\newtheorem{notation}[theorem]{Notation}
\newtheorem{remark}[theorem]{Remark}
\newtheorem{remarks}[theorem]{Remarks}
\newtheorem{definition}[theorem]{Definition}
\newtheorem{claim}[theorem]{Claim}
\newtheorem{assumption}[theorem]{Assumption}
\numberwithin{equation}{section}
\newtheorem{example}[theorem]{Example}
\newtheorem{examples}[theorem]{Examples}
\newtheorem{propositionrm}[theorem]{Proposition}

\newcommand{\binomial}[2]{\left(\begin{array}{c}#1\\#2\end{array}\right)}
\newcommand{\zar}{{\rm zar}}
\newcommand{\an}{{\rm an}}
\newcommand{\red}{{\rm red}}
\newcommand{\codim}{{\rm codim}}
\newcommand{\rank}{{\rm rank}}
\newcommand{\Pic}{{\rm Pic}}
\newcommand{\Div}{{\rm Div}}
\newcommand{\Hom}{{\rm Hom}}
\newcommand{\im}{{\rm im}}
\newcommand{\Spec}{{\rm Spec}}
\newcommand{\sing}{{\rm sing}}
\newcommand{\reg}{{\rm reg}}
\newcommand{\Char}{{\rm char}}
\newcommand{\Tr}{{\rm Tr}}
\newcommand{\res}{{\rm res}}
\newcommand{\tr}{{\rm tr}}
\newcommand{\supp}{{\rm supp}}
\newcommand{\Gal}{{\rm Gal}}
\newcommand{\Min}{{\rm Min \ }}
\newcommand{\Max}{{\rm Max \ }}
\newcommand{\Span}{{\rm Span  }}

\newcommand{\Frob}{{\rm Frob}}
\newcommand{\lcm}{{\rm lcm}}


\long\def\symbolfootnote[#1]#2{\begingroup%
\def\thefootnote{\fnsymbol{footnote}}\footnote[#1]{#2}\endgroup}

\newcommand{\soplus}[1]{\stackrel{#1}{\oplus}}
\newcommand{\dlog}{{\rm dlog}\,}    
\newcommand{\limdir}[1]{{\displaystyle{\mathop{\rm
lim}_{\buildrel\longrightarrow\over{#1}}}}\,}
\newcommand{\liminv}[1]{{\displaystyle{\mathop{\rm
lim}_{\buildrel\longleftarrow\over{#1}}}}\,}
\newcommand{\boxtensor}{{\Box\kern-9.03pt\raise1.42pt\hbox{$\times$}}}
\newcommand{\sext}{\mbox{${\mathcal E}xt\,$}}
\newcommand{\shom}{\mbox{${\mathcal H}om\,$}}
\newcommand{\coker}{{\rm coker}\,}
\renewcommand{\iff}{\mbox{ $\Longleftrightarrow$ }}
\newcommand{\onto}{\mbox{$\,\>>>\hspace{-.5cm}\to\hspace{.15cm}$}}

\newenvironment{pf}{\noindent\textbf{Proof.}\quad}{\hfill{$\Box$}}

\newcommand{\sA}{{\mathcal A}}
\newcommand{\sB}{{\mathcal B}}
\newcommand{\sC}{{\mathcal C}}
\newcommand{\sD}{{\mathcal D}}
\newcommand{\sE}{{\mathcal E}}
\newcommand{\sF}{{\mathcal F}}
\newcommand{\sG}{{\mathcal G}}
\newcommand{\sH}{{\mathcal H}}
\newcommand{\sI}{{\mathcal I}}
\newcommand{\sJ}{{\mathcal J}}
\newcommand{\sK}{{\mathcal K}}
\newcommand{\sL}{{\mathcal L}}
\newcommand{\sM}{{\mathcal M}}
\newcommand{\sN}{{\mathcal N}}
\newcommand{\sO}{{\mathcal O}}
\newcommand{\sP}{{\mathcal P}}
\newcommand{\sQ}{{\mathcal Q}}
\newcommand{\sR}{{\mathcal R}}
\newcommand{\sS}{{\mathcal S}}
\newcommand{\sT}{{\mathcal T}}
\newcommand{\sU}{{\mathcal U}}
\newcommand{\sV}{{\mathcal V}}
\newcommand{\sW}{{\mathcal W}}
\newcommand{\sX}{{\mathcal X}}
\newcommand{\sY}{{\mathcal Y}}
\newcommand{\sZ}{{\mathcal Z}}

\newcommand{\A}{{\mathbb A}}
\newcommand{\B}{{\mathbb B}}
\newcommand{\C}{{\mathbb C}}
\newcommand{\D}{{\mathbb D}}
\newcommand{\E}{{\mathbb E}}
\newcommand{\F}{{\mathbb F}}
\newcommand{\G}{{\mathbb G}}
\newcommand{\HH}{{\mathbb H}}
\newcommand{\I}{{\mathbb I}}
\newcommand{\J}{{\mathbb J}}
\newcommand{\M}{{\mathbb M}}
\newcommand{\N}{{\mathbb N}}
\renewcommand{\P}{{\mathbb P}}
\newcommand{\Q}{{\mathbb Q}}
\newcommand{\T}{{\mathbb T}}
\newcommand{\U}{{\mathbb U}}
\newcommand{\V}{{\mathbb V}}
\newcommand{\W}{{\mathbb W}}
\newcommand{\X}{{\mathbb X}}
\newcommand{\Y}{{\mathbb Y}}
\newcommand{\Z}{{\mathbb Z}}


\newcommand{\Fqm}{\mathbb{F}_{q^m}}
\newcommand{\Fq}{\mathbb{F}_q}
\newcommand{\Fp}{\mathbb{F}_p}
\newcommand{\Fpl}{\mathbb{F}_{p^l}}
\newcommand{\fqn}{\mathbb{F}_q^n}
\newcommand{\be}{\begin{eqnarray}}
\newcommand{\ee}{\end{eqnarray}}
\newcommand{\nn}{{\nonumber}}
\newcommand{\dd}{\displaystyle}
\newcommand{\ra}{\rightarrow}
\newcommand{\bigmid}[1][12]{\mathrel{\left| \rule{0pt}{#1pt}\right.}}
\newcommand{\cl}{${\rm \ell}$}
\newcommand{\clp}{${\rm \ell^\prime}$}

\title[Self-dual codes]{Self-Dual Codes better than the Gilbert--Varshamov bound}
\author{Alp Bassa \and Henning Stichtenoth}
\thanks{A.B. was supported by the BAGEP Award of the Science Academy with funding supplied by Mehve{\c s} Demiren in memory of Selim Demiren and T\"UB\.ITAK Proj. 112T233. H.S. was supported by T\"UB\.ITAK Proj. 114F432.}
\maketitle

\begin{abstract} We show that every self-orthogonal code over $\mathbb F_q$ of length $n$ can be extended to a self-dual code, if there exists self-dual codes of length $n$. Using a family of Galois towers of algebraic function fields we show that over any nonprime field $\mathbb F_q$, with $q\geq 64$, except possibly $q=125$, there are self-dual codes better than the asymptotic Gilbert--Varshamov bound.
\end{abstract}

\section{Introduction}
\label{sec1}

Let $\Fq$ be the finite field of cardinality $q$ where $q$ is a power of some prime
number $p$. We mean by a code $C$ over $\Fq$ always a  linear code; i.e.,
$C$ is a linear subspace of the $n$-dimensional vector space $\fqn$. The number $n$ is 
called the length of $C$, and the dimension $k$ of $C$ as an $\Fq$-vector space is called the dimension of $C$. The weight of an element $x=(x_1,\ldots,x_n)\in \fqn$ is defined as 
$${\rm wt}(x)=|\{i:x_i\neq 0\}|.$$
The minimum distance $d$ of a code $C$ is defined as 
$$d=\min_{0\neq x\in C} {\rm wt}(x).$$

\noindent The space $\fqn$ is equipped with the standard symmetric bilinear form 
$\left\langle \cdot , \cdot \right\rangle$, defined by
$$ \left\langle x,y\right\rangle = \sum_{i=1}^n x_iy_i \ $$
for $x=(x_1,\ldots,x_n), y=(y_1,\ldots,y_n) \in \fqn$. Clearly this bilinear form is non-degenerate;
i.e., for every $0 \ne x \in \fqn$ there is some $y \in \fqn$ such that $\left\langle x,y\right\rangle \neq 0$.
 Two vectors $u,v \in \fqn$ are called orthogonal if $\left\langle u,v \right\rangle =0$. In this case we 
 also write $u \bot v$. If $C \subseteq \fqn$ is a code then the set
 $$ C^{\bot} : = \{ y \in \fqn \ | \ y \bot x \ {\hbox {\rm for all}} \ x \in C \}
$$
is also a linear subspace of $\fqn$; it is called the dual code of $C$. It is clear from linear algebra
that 
\begin{equation}\label{1.1}
\dim C + \dim C^{\bot} = n \ .
\end{equation}
A code $C$ is called self-orthogonal if $C \subseteq C^{\bot}$. It is called self-dual if $C=C^{\bot}$.
An $[n,k,d]$-code $C$ is a code of length $n$, dimension $k$ and minimum distance $d$. The ratios 
$$ R(C) := k/n \ \ {\hbox {\rm and}} \ \ \delta(C) := d/n $$
are called the rate and the relative minimum distance of $C$, resp. It is clear from Equation ({\ref{1.1}})
that the rate of a self-orthogonal code $C$ satisfies $R(C) \le 1/2$; a self-dual code has rate 
$R(C) = 1/2$.

It has been known for a long time that the class of self-dual codes over $\Fq$ is asymptotically good 
and it attains the Gilbert--Varshamov bound (\cite{MWST,PP}). This means: there exists
a sequence $(C_i)_{i \ge 0}$ of self-dual codes over $\Fq$ with length $n_i \to \infty$ such that
the limit $\delta : = \lim_{i \to \infty} \delta (C_i)$  exists and the point $(\delta , R) \in \mathbb R ^2$
with $R=1/2$
lies on or above the Gilbert--Varshamov bound
\begin{equation}\label{1.2}
 R \ge 1 - H_q(\delta) \ .
\end{equation}
Here $H_q(\delta)$ denotes the $q$-ary entropy function, defined by $H_q(0) = 0$ and
$$ H_q(\delta) = \delta \log_q (q-1) - \delta \log_q (\delta) - (1-\delta)\log_q (1-\delta)
$$
for $0 < \delta \le 1 - q^{-1}$. 

For $q=\ell^2$ a square, the bound \eqref{1.2} was improved in \cite{trans}. More precisely it was shown that the class of self-dual codes attains the Tsfasman--Vladut--Zink bound; i.e., there is a sequence of self-dual codes $(C_i)_{i\geq0}$ over $\mathbb F_q$ with parameters $[n_i,n_i/2,d_i]$ with $n_i\to \infty$ and
\begin{equation} \label{transeq}
\liminf_{i\to \infty} \frac{d_i}{n_i}\geq \frac{1}{2} - \frac{1}{\ell -1}.
\end{equation}

Our aim is to extend \eqref{transeq} to all nonprime finite fields and hence to improve the bound \eqref{1.2} for almost all nonprime values of $q$ (all, except $q\leq 49$ and $q=125$).

Our proof relies on the use of specific towers of algebraic function fields having many rational 
places. These towers will allow us to construct self-orthogonal algebraic geometry codes $C_i$ of increasing
 length whose dual codes 
$C_i^{\bot}$ have a large minimum distance. We will then show that there  are self-dual codes 
$\widetilde{C}_i$ with $C_i \subseteq \widetilde{C}_i \subseteq C_i^{\bot}$ whose relative minimum distance satisfies the 
corresponding Tsfasman-Vladut-Zink bound 
\begin{equation}\label{1.3}
 \liminf_{i \to \infty}\delta(\widetilde{C}_i) \ge  \frac{1}{2}-\frac{1}{2}\Bigl(\frac{1}{\ell^{\lceil r/2\rceil}-1}+\frac{1}{\ell^{\lfloor r/2\rfloor}-1}\Bigr) \quad \text{ for } q = \ell^r, r>1 \text{ odd.}
 \end{equation}

Taking $r=2$ in \eqref{1.3}, we recover \eqref{transeq}. Together, \eqref{transeq} and \eqref{1.3} give the following result
\begin{theorem} For any nonprime finite field $\mathbb F_q$ with $q\geq 64$, except possibly $q=125$, there are self-dual codes over $\mathbb F_q$ better than the Gilbert--Varshamov bound.
\end{theorem}

\section{Embedding Self-Orthogonal Codes into Self-Dual Codes} \label{sec2}

In this section we show that every self-orthogonal code $C \subseteq \fqn$ can be
extended to a self-dual code $\widetilde{C} \subseteq \fqn$, if at least one self-dual code exists in $\fqn$. 
The results in this section should be classically known in the theory of quadratic spaces over finite fields and finite geometries. In particular, Lemma~\ref{lemma2.2} follows immediately from the classification of quadratic spaces over finite fields by dimension and discriminant and can be found in \cite[Theorem 1]{PlessGolay}, and results along the lines of  Theorem~\ref{thm2.1} for $q=2$ can be found in \cite{MWST}, among possibly others. Since we could not find an good reference in this form and generality, we give proofs of these results below. 

\begin{theorem}\label{thm2.1}
Let $C \subseteq \fqn$ be a self-orthogonal code over $\Fq$ of length $n$. Assume that
\begin{itemize}
\item[($\star$)] $n$ is even and, in case $q \equiv 3 \pmod 4$, $n$ is a multiple of $4$.
\end{itemize}
Then there exists a self-dual code $\widetilde{C} \subseteq \fqn$ such that $C \subseteq \widetilde{C}$.
\end{theorem}

Condition ($\star$) above is necessary and sufficient for the existence of a self-dual code of length $n$ over $\Fq$. Necessity can be seen easily by using the discriminant of the bilinear form $\langle\cdot,\cdot \rangle$, but we do not need it here. To show sufficiency (to show that under the
above condition ($\star$) at least one self-dual code over $\Fq $ of length $n$ exists) constitutes the first step for
the proof of Theorem \ref{thm2.1}:

\begin{lemma}\label{lemma2.2}
Assume that $n$ is even and, in case $q \equiv 3 \pmod 4$, $n$ is a multiple of $4$. Then there exists a self-dual code $E \subseteq \fqn$.
\end{lemma}
\begin{proof} First we consider case that $q$ is even or $q \equiv 1 \pmod{4}$ and $n=2m$ is even. If $q$ is even, let $\alpha=1$. If 
$q \equiv 1 \pmod 4$, since the multiplicative group $\Fq^{\times}$ is cyclic
of order $q-1 \equiv 0 \pmod 4$, there is an element $\alpha \in \Fq^{\times}$ such that
$\alpha^2 = -1$. We consider the vectors
$$
c_1 = (\alpha, 1, 0,0, \ldots, 0,0), \ \ c_2 = (0, 0, \alpha, 1, \ldots, 0,0),\ \ldots \ ,\
c_m = (0, 0,  \ldots , 0, 0 ,  \alpha, 1).$$
These vectors span an $m$-dimensional subspace $E \subseteq \fqn$ which is obviously self-dual.

Next we consider the case $q \equiv 3 \pmod 4$ and $n = 4k$. Since $\Fq^{\times}$ has order $q-1 \equiv 2 \pmod 4$,
the element $-1 \in \Fq$ is not a square. So the set
$$
A : = \{ 1 + \alpha^2 \ | \ \alpha \in \Fq \} \subseteq \Fq^{\times}
$$
has cardinality $|A| = (q+1)/2$. Let
$$
B : = \{ - \beta^2 \ | \ \beta \in \Fq^{\times} \} \subseteq \Fq^{\times} \ ,
$$
then $|B| = (q-1)/2$, so the intersection $A \cap B$ is non-empty. Therefore we find elements
$\alpha , \beta \in \Fq$ satisfying $\alpha^2 + \beta^2 + 1 = 0$. We consider the vectors
\begin{eqnarray*} c_1 &=& (\alpha, \beta, 1, 0, 0,0,0,0, \ldots ,0,0,0,0), \ c_2 = (0,0,0,0,\alpha, \beta, 1, 0,\ldots ,0,0,0,0), \ldots , \\
c_k &=& (0,0,0,0, \cdots , 0,0,0,0,\alpha, \beta, 1, 0) \end{eqnarray*}
and
\begin{eqnarray*} d_1 &=& (-\beta, \alpha, 0, 1, 0,0,0,0, \ldots , 0,0,0,0), \ d_2 = (0,0,0,0,-\beta, \alpha, 0, 1,  \ldots , 0,0,0,0), \ldots , \\
d_k &=& (0,0,0,0, \cdots ,0,0,0,0,-\beta, \alpha, 0, 1) \ .\end{eqnarray*}
Now the vectors $c_1, d_1, \ldots, c_k,d_k$ span a self-dual code $E \subseteq \fqn$.
\end{proof}
In order to prove  
Theorem \ref{thm2.1} we use Witt's 
Theorem which holds in a more general setting as follows. 
Let $K$ be an arbitrary field and let $V$ be a 
vector space over $K$. Let $s: V\times V \longrightarrow K$ be a symmetric bilinear form, and
let $W \subseteq V$ be a subspace of $V$.
 An injective linear map
$\varphi : W \longrightarrow V$ is called an isometry from $W $ to $V$ if $\left\langle \varphi (w_1),
\varphi (w_2)\right\rangle = \left\langle w_1,w_2\right\rangle$ holds for all $w_1,w_2 \in W$. Now we
can state Witt's Theorem.

\begin{theorem}[Witt \cite{Se}]\label{thm2.3} Let  $V$ be a vector space over $K$ with 
$\dim V = n < \infty$, where $K$ is a field of characteristic ${\rm char} \ K \ne 2$. Let $s$ be a
non-degenerate symmetric bilinear form on $V$ and let $W \subseteq V$ be a subspace of $V$. Assume that
$\varphi:W \longrightarrow V$ is an isometry. Then $\varphi$ can be extended to an isometry $\widetilde {\varphi}
: V \longrightarrow V$; i.e., $\varphi$ is the restriction of $\widetilde{\varphi}$ to $W$.
\end{theorem}

Pless has shown that under certain conditions, an analog of Witt's Theorem holds in characteristic $2$. In the particular case, where $K$ a finite field of characteristic $2$, it gives the following:
\begin{theorem}[Pless \cite{PlessWitt}]\label{thm2.4}  Let $\boldsymbol{1}=(1,1,\ldots,1)$. Consider the setting of Theorem~\ref{thm2.3}, where $K$ is a finite field of characteristic $2$. Assume moreover that the following holds:
If $\boldsymbol{1}\in W$, then $\varphi(\boldsymbol{1})=\boldsymbol{1}$. Otherwise, if $\boldsymbol{1}\notin W$, then $\boldsymbol{1}\notin \varphi(W)$. Then the conclusion of Theorem~\ref{thm2.3} holds.
\end{theorem}

\begin{proof}[Proof of Theorem {\rm \ref{thm2.1}}] We are given a self-orthogonal code $C \in \fqn$.
By Lemma \ref{lemma2.2} there exists
a self-dual code $E \subseteq \fqn$.  As $\dim  C \le \dim  E$, there is an injective linear map
$\varphi : C \longrightarrow E$ (if $q$ is even, choose $\varphi$ so that the condition in Theorem~\ref{thm2.4} is satisfied).  Since both codes $C$ and $E$ are self-orthogonal, it follows that
$\varphi$ is in fact an isometry, and we can extend $\varphi$ to an isometry $\widetilde{\varphi} : \fqn
\longrightarrow \fqn$ by Witt's Theorem (or Pless' Theorem if $q$ is even). Then the space $\widetilde{C} : = \widetilde{\varphi}^{-1} (E)$
is a self-dual code (as $\widetilde{\varphi}^{-1}$ is an isometry) and it contains $C$.
\end{proof}

\section{Self-dual algebraic geometry codes}\label{sec3}
\label{secion:goppa}
Let us first fix some notation. For background on the theory of algebraic function fields, we refer to \cite{Sti}.
We will consider function fields $F/\mathbb F_q$ where $\mathbb F_q$ is the full constant field of $F$. We will denote by

\begin{tabular}{ll}
 $g(F)$ & the genus of $F$,\\
 $(x)$ & the principal divisor of $0\neq x\in F$,\\
 $x(P)$ & the value of the function $x\in F$ at the place $P\in \mathbb P_F$,\\
 $\mathbb P_F$ & the set of places of $F/ \mathbb F_q$,\\
 $v_P$ & the normalized discrete valuation of $F/\mathbb F_q$ associated with the \\ & place $P\in \mathbb P_F$,\\
 $N(F)$ & the number of places of degree one (rational places) of $F/\mathbb F_q$,\\
 ${\rm supp} A$ & the support of the divisor $A$ of $F/ \mathbb F_q$,\\
 $(\omega)$ & the divisor of the differential $\omega \neq 0$,\\
 ${\rm res}_P(\omega)$ & the residue of the differential $\omega$ at the place $P\in \mathbb P_F$.\\
\end{tabular}

For a divisor $A$ of $F/\mathbb F_q$ we define the Riemann-Roch space
$$  L(A):=\{x\in F^\times | (x)+A \geq 0\} \cup \{0\}.$$


For a finite separable extension $E$ of $F$ we will denote by 

\begin{tabular}{ll}
 ${\rm Con}_{E/F}(A)$ & the conorm of the divisor $A$ of $F$ in $E/F$,\\
 ${\rm Cotr}_{E/F}(\omega)$ & the cotrace of the differential $\omega$ of $F$ in $E/F$,\\
 ${\rm Diff}(E/F)$ & the different of the extension $E/F$.\\
\end{tabular}

A rational place $P\in \mathbb P_F$ is said to split completely in the extension $E/F$ if there are $[E:F]$ distinct places of $E$ above $P$.

Let $F/ \mathbb F_q$ be an algebraic function field of genus $g$ over the finite field $\mathbb F_q$. Let $P_1, P_2, \ldots P_n$ be pairwise different rational places of $F/\mathbb F_q$. Put $D=P_1+P_2+\ldots+P_n$, and let $G$ be a divisor of $F/\mathbb F_q$, such that ${\rm supp}\, G \cap {\rm supp}\, D=\varnothing$. We consider the algebraic geometry code $  C_{  L}(G,D)$, which is as usual defined as follows:
$$  C_{  L}(G,D):=\{(f(P_1),f(P_2),\ldots,f(P_n))|f\in   L(G)\}.$$
It is well known that this is a linear code of length $n$ and minimum distance $d$, with
$$d \geq n-\deg G\quad (\text{if }   C_{  L}(G,D)\neq 0).$$

In \cite{self}, sufficient criteria for self-duality of algebraic geometry codes are given. In particular, we have the following description of the dual $ C_{  L}(G,D)^\perp$ of the code $C_{  L}(G,D)$:
\begin{theorem}
\label{thm:orthogonal}
Suppose $\omega$ is a differential such that 
\begin{enumerate}
\item $v_{P_i}(\omega)=-1$, for $i=1, 2,\ldots, n$,
\item ${\rm res}_{P_i}(\omega)= {\rm res}_{P_j}(\omega)$ for $1\leq i,j\leq n$.
\end{enumerate}
Then we have 
$$  C_{  L}(G,D)^\perp=  C_{  L}(D+(\omega)-G,D).$$
\end{theorem}
\begin{proof}
See \cite{self}.
\end{proof}

\begin{corollary} 
\label{cor2}
(in the setting as above) If $D+(\omega)\geq 2G$ then $  C_{  L}(G,D)\subseteq   C_{  L}(G,D)^\perp$.
\end{corollary} 
\begin{proof}
\begin{align*}
G \leq D+(\omega)-G &\Rightarrow   L(G) \subseteq   L(D+(\omega)-G)\\
 &\Rightarrow   C_{  L}(G,D) \subseteq   C_{  L}(D+(\omega)-G,D)=C_{  L}(G,D)^\perp.
 \end{align*}
\end{proof}
\begin{corollary}
\label{cor3} \label{corself}
(in the setting as above) If $D+(\omega)=2G$ then $  C_{  L}(G,D)=  C_{  L}(G,D)^\perp$; i.e., the code $  C_{  L}(G,D)$ is self-dual.
\end{corollary}

If $v_P((\omega)+D)$ is even for all $P\in {\rm supp}((\omega)+D)$, we can obtain self-dual algebraic geometry codes by taking $G=\frac{D+(\omega)}{2}$. Otherwise, we can construct self-dual codes using Theorem~\ref{thm2.1} as follows:

We will call a divisor $A$ even if $v_p(A)$ is even for all $P\in \mathbb P_F$. For a divisor $A$ we will denote by $\lfloor A \rfloor$ (respectively $\lceil A \rceil$) the largest (respecectively smallest) even divisor $B$ with $B\leq A$ (respectively $B\geq A$). Clearly $2A=\lfloor A \rfloor+\lceil A \rceil$.


\begin{theorem} \label{thm:codeconstr}
Let $n$ be an even integer, that is also a multiple of $4$ in case $q \equiv 3 \pmod 4$. Let $P_1, P_2, \ldots P_n$ be pairwise different rational places of $F/\mathbb F_q$, let $D=P_1+P_2+\ldots+P_n$ and let $\omega$ be a differential, such that 
\begin{enumerate}
\item $v_{P_i}(\omega)=-1$, for $i=1, 2,\ldots, n$,
\item ${\rm res}_{P_i}(\omega)= {\rm res}_{P_j}(\omega)$ for $1\leq i,j\leq n$.
\end{enumerate}
Then there exists a self-dual code of length $n$ and minimum distance $d$ satisfying
$$d \geq \frac{\deg(\lfloor D+(\omega)\rfloor)}{2}-\deg (\omega).$$
\end{theorem}
\begin{proof}
Let $G=\frac{\lfloor (\omega)+D \rfloor}{2}$. Since $2G=\lfloor (\omega)+D \rfloor \leq (\omega)+D$, it follows from  Corollary~\ref{cor2} that the code $C_{L}(G,D)$ is self-orthogonal. By Theorem~\ref{thm:orthogonal} we have
$$  C_{  L}(G,D)^\perp=  C_{L}(D+(\omega)-G,D)=  C_{  L}(\frac{\lceil D+(\omega)\rceil}{2},D).$$
Hence for the minimum distance of $  C_{  L}(G,D)^\perp$ we obtain the estimate
\begin{eqnarray*}
d(  C_{  L}(G,D)^\perp)&=&d(  C_{  L}(\frac{\lceil D+(\omega)\rceil}{2},D))\geq n-\deg(\frac{\lceil D+(\omega)\rceil}{2})\\
&=&\deg(\frac{\lfloor D+(\omega)\rfloor}{2})-\deg (\omega)\\
&=&\frac{\deg(\lfloor D+(\omega)\rfloor)}{2}-\deg (\omega).
\end{eqnarray*}
By Theorem~\ref{thm2.1} we see that there is a self-dual code $\widetilde{C}$ with $C_{L}(G,D) \subseteq \widetilde{C} \subseteq  C_{  L}(G,D)^\perp$. From this inclusion we obtain for the minimum distance $d(\widetilde{C})$ of $\widetilde{C}$
$$d(\widetilde{C})\geq d(C_{  L}(G,D)^\perp)\geq \frac{\deg(\lfloor D+(\omega)\rfloor)}{2}-\deg (\omega).$$
\end{proof}

\section{Asymptotically good self-dual codes} \label{sec4}
The real strength of algebraic geometry codes becomes apparent when considering asymptotic questions, i.e., families of codes of increasing length. The length of an algebraic geometry code is limited by the number of rational places $N(F)$ of the function field $F$. Hence to consider codes of increasing length one is naturally led to work with function fields with many rational places, which will necessarily have large genera. Thus let us briefly recall the notion of a tower of function fields. 

A tower $\mathcal  F$ of function fields over $\mathbb F_q$ is an infinite sequence $\mathcal F=(F_0, F_1, F_2, \ldots)$ of function fields $F_i / \mathbb F_q$, with the following properties:
\begin{enumerate}
\item $F_0 \subseteq F_1 \subseteq F_2 \subseteq \ldots$.
\item The field $\mathbb F_q$ is the full constant field of $F_i$, for $i=0, 1, 2,\ldots$.
\item For each $i \geq 1$, the extension $F_i / F_{i-1}$ is finite and separable.
\item $g(F_i) \to \infty$ as $i \to \infty$.
\end{enumerate}

A tower $\mathcal  F=(F_0, F_1, F_2, \ldots)$ is called a Galois tower, if all extensions $F_i/F_0$ are Galois. For a Galois tower $  \mathcal F=(F_0, F_1, F_2, \ldots)$, a place $P\in \mathbb P_{F_0}$ and a place $Q\in \mathbb P_{F_i}$, we will denote by $e_i(P)$ the ramification index $e(Q|P)$ of $Q|P$. Note that since all extensions $F_i/F_0$ are Galois, $e_i(P)$ is well-defined; i.e., does not depend on the chosen place $Q$ of $F_i$ lying over $P$. 

We define the genus $\gamma(\mathcal  F/F_0)$ of $\mathcal  F$ over $F_0$
$$\gamma(\mathcal  F):=\lim_{i \to \infty} \frac{g(F_i)}{[F_i : F_0 ]}.$$
It can be shown, that this limit exists (it can be $\infty$).

A place $P$ of $F_0$ is said to be ramified in the tower $\mathcal  F=(F_0, F_1, F_2, \ldots)$, if the place $P$ is ramified in the extension $F_i / F_0$ for some $i\geq 1$. 
The set 
$$V(\mathcal  F / F_0):=\{ P \in \mathbb P({F_0}) | P \textrm{ is ramified in }   F \} $$
is called the ramification locus of $\mathcal  F$ over $F_0$. 

A rational place $P$ of $F_0$ is said to split completely in the tower $\mathcal  F=(F_0, F_1, F_2, \ldots)$, if the place $P$ splits completely in all extensions $F_i / F_0$. 

\begin{theorem} \label{constr} Suppose there exists a tower $\mathcal  F=(F_0, F_1, \ldots)$ of function fields $F_i / \mathbb F_q$  satisfying the following conditions:
\begin{enumerate}
\item The extension $F_i/F_0$ is Galois for every $i\geq 1$,
\item the ramification locus $V( \mathcal F/F_0)$ of the tower is finite. Moreover for any place $P\in V(\mathcal  F/F_0)$, we have $\lim_{i \to \infty} e_i(P)=\infty$, where $e_i(P)$ denotes the ramification index of the place $P$ in the extension $F_i/F_0$,
\item there exists a differential $\omega$ of $F_0$, such that:
\begin{itemize}
\item ${\rm supp}((\omega))=\{R_0, R_1, \ldots, R_k\} \cup \{S_0, S_1, \ldots, S_m\}\subseteq \mathbb P_{F_0}$,
\item $m>0$, the places $S_0, S_1, \ldots, S_m$ are rational and split completely in the tower, moreover we have $v_{S_j}(\omega)=-1$ and ${\rm res}_{S_j}(\omega)=1$ for $0 \leq j \leq m$,
\item the places $R_0, R_1, \ldots, R_k$ are ramified in the tower,
\end{itemize}
\item  
\begin{itemize} 
\item if $q \equiv  3 \pmod{4}$ then $4 | [F_r:F_0]\cdot m$ for some $r\geq 1$
\item if $q$ is even or $q \equiv 1 \pmod{4}$ then $2 | [F_r:F_0]\cdot m$ for some $r\geq 1$
\end{itemize}
\end{enumerate}
Then there exists a sequence $(C_i)_{i \geq 0}$ of self-dual codes over $\mathbb F_q$, such that
$$n(C_i) \to \infty \quad \text{and} \quad \liminf_{i \to \infty} \frac{d_i}{n_i} \geq \frac{1}{2}-\frac{\gamma(\mathcal  F/F_0)}{m},$$
where 
$$\gamma(\mathcal  F/F_0)=\lim_{i \to \infty} \frac{g(F_i)}{[F_i:F_0]}$$
denotes the genus of the tower $\mathcal  F$.
\end{theorem}
\begin{proof}
For $i\geq 0$ consider the differential $\omega_i={\rm Cotr}_{F_i/F_0}(\omega)$ of $F_i$. Let $D_i={\rm Con}_{F_i/F_0} (S_0+S_1+\ldots+S_m)$. Since for $0 \leq j \leq m$ the place $S_j$ splits completely in the tower and since $v_{S_j}(\omega)=-1$ and ${\rm res} _{S_j}(\omega)=1$, we have $v_{Q}(\omega_i)=-1$ and ${\rm res}_{Q}(\omega_i)=1$ for any $Q \in {\rm supp}(D_i)$. Without loss of generality we can assume that $i\geq r$. Since $|{\rm supp}D_i|=m\cdot [F_i:F_0]$ it follows hence by Theorem~\ref{thm:codeconstr} that there exists a self-dual code $C_i$ of length $n_i=m\cdot [F_i:F_0]$ and minimum distance $d_i$ satisfying
$$d_i \geq \frac{\deg(\lfloor D_i+(\omega_i)\rfloor)}{2}-\deg (\omega_i).$$
Next we want to estimate $\deg(\lfloor D_i+(\omega_i)\rfloor)$. Let $T:={\rm supp}(D_i+(\omega_i))$. Clearly we have
$$\deg(\lfloor D_i+(\omega_i)\rfloor)\geq \deg(D_i+(\omega_i))-\sum_{Q\in T} \deg Q.$$
Since $(\omega_i)={\rm Con}_{F_i/F_0}((\omega))+{\rm Diff}(F_i/F_0)$, it follows that every place in the support of $D_i+(\omega_i)$ lies over a place in the ramification locus $V(\mathcal  F/F_0)$ of $\mathcal  F$ (which is finite!). Hence
$$\sum_{Q\in T} \deg Q \leq \sum_{P\in V(\mathcal  F/F_0)} \frac{[F_i:F_0]}{e_i(P)}\deg P.$$
Hence
\begin{eqnarray*}
d_i &\geq& \frac{\deg(\lfloor D_i+(\omega_i)\rfloor)}{2}-\deg (\omega_i) \\
&\geq& \frac{\deg(D_i+(\omega_i))-\sum_{Q\in T} \deg Q}{2} - \deg (\omega_i) \\
&\geq& \Bigl( \deg(D_i)-\deg(\omega_i)-\sum_{P\in V(\mathcal  F/F_0)} \frac{[F_i:F_0]}{e_i(P)}\deg P\Bigr) /2.
\end{eqnarray*}
Dividing by $n_i=m\cdot [F_i:F_0]=\deg(D_i)$ and using $\deg(\omega_i)=2g(F_i)-2$, we obtain
$$\frac{d_i}{n_i}\geq \frac{1}{2}-\frac{g(F_i)}{m\cdot [F_i:F_0]}+\frac{1}{m\cdot [F_i:F_0]}-\Bigl(\sum_{P\in V( \mathcal F/F_0)} \frac{1}{e_i(P)}\deg P\Bigr)/(2m).$$
Letting $i\to \infty$ and noting that for all $P\in V( \mathcal F/F_0)$ we have $\lim_{i\to \infty} e_i(P)=\infty$,  we obtain the desired result.
\end{proof}

Galois towers over all non-prime finite fields were constructed in \cite{BBGS}. In particular in \cite[Theorem 1.1]{BBGS} it is shown that 
for $\ell$ a prime power, $q=\ell^r$ with $r>1$ odd there exists  a tower $\mathcal  F=(F_0, F_1, F_2, \ldots)$ of functions fields over $\mathbb F_q$, satisfying the conditions in Theorem~\ref{constr} with $m=1$ and 
\begin{equation}\label{ineqodd}
\gamma(\mathcal F/F_0)\geq \frac{1}{2}\Bigl(\frac{1}{\ell^{(r-1)/2}-1}+\frac{1}{\ell^{(r+1)/2}-1}\Bigr)
\end{equation}
Note that in the corresponding tower the field $F_0=\mathbb F_q(z)$ is a rational function field and we take 
$$\omega:=\frac{{\rm d}z}{z-1}.$$
Then all conditions of Theorem~\ref{constr} are easily verified.

\begin{remark} \label{remarklabel}
Let $q$ be a prime power, that is not a prime, with $q\geq 64$ and $q\neq 125$. Let $\delta_0$ be such that $1-H_q(\delta_0)=1/2$. Then there is a prime power $\ell$ and an integer $r>1$ such that $q=\ell^r$ and 
$$\delta_0<\frac{1}{2}-\frac{1}{2}\Bigl(\frac{1}{\ell^{\lceil r/2\rceil}-1}+\frac{1}{\ell^{\lfloor r/2\rfloor}-1}\Bigr)=:\delta_1$$
\end{remark}
\begin{proof}[Proof of Remark \ref{remarklabel}]
Since $1-H_q(\delta)$ is a strictly decreasing function, it is sufficient to show that 
\begin{equation} \label{ineq2}
1-H_q(\delta_1)<1/2.
\end{equation} 
Let 
$$\epsilon=\frac{1}{2}\Bigl(\frac{1}{\ell^{\lceil r/2\rceil}-1}+\frac{1}{\ell^{\lfloor r/2\rfloor}-1}\Bigr).$$
Using the Taylor series expansion of $\log_q(1+x)$, we see that 
$$\log_q(q-1)>1-\frac{1}{\ln(q)}\cdot\frac{1}{q-1},\ \log_q(\frac{1}{2}-\epsilon)< \frac{1}{\ln(q)}\cdot\bigl(-\frac{1}{2}-\epsilon\bigr),
\ \log_q\bigl(1-(\frac{1}{2}-\epsilon)\bigr)<\frac{1}{\ln(q)}\cdot\bigl(-\frac{1}{2}+\epsilon\bigr).$$
Hence to show Inequality~\eqref{ineq2} it is sufficient to show that 
$$\bigl(\frac{1}{2}-\epsilon\bigr)\Bigl[1-\frac{1}{\ln(q)}\cdot\frac{1}{q-1}+\frac{1}{\ln(q)}\cdot\bigl(1+2\epsilon\bigr)\Bigr]>\frac{1}{2}.$$
Noting that $1/\epsilon$ is the harmonic mean of $\ell^{\lceil r/2\rceil}-1$ and $\ell^{\lfloor r/2\rfloor}-1$ and therefore 
$2\epsilon\geq 1/(q-1)$, it is enough to show that 
$$\frac{1}{\epsilon}>2+2\ln(q).$$
For the same reason, $1/\epsilon \geq \ell^{\lfloor r/2\rfloor}-1$, so it suffices to show $\ell^{\lfloor r/2\rfloor}>3+2\ln(l^r)$. This inequality is easily checked for $\ell\geq 23$ or $r>7$. Direct calculation in the finitely many remaining cases shows that Inequality~\eqref{ineq2} holds except for $\ell^r \leq 49$, $(l,r)=(5,3)$ and $(l,r)=(4,3)$ (it does however hold for $(l,r)=(2,6)$). The result follows.
\end{proof}
Inequalities \eqref{transeq} and \eqref{ineqodd} and Remark~\ref{remarklabel} together yield the following result over all nonprime finite fields:
\begin{theorem} 
Let $q=\ell^r$ with $r>1$. There exists a sequence $(C_i)_{i\geq0}$ of self-dual codes over $\mathbb F_q$ having parameters $[n_i, n_i/2, d_i]$ with $n_i \to \infty$ and
$$
\liminf_{i \to \infty} d_i/n_i \geq \frac{1}{2}-\frac{1}{2}\Bigl(\frac{1}{\ell^{\lceil r/2\rceil}-1}+\frac{1}{\ell^{\lfloor r/2\rfloor}-1}\Bigr).
$$
Hence for all nonprime $q$ with $q\geq 64$ except $q=125$ there are self-dual codes better than the Gilbert--Varshamov bound. 
\end{theorem}

\begin{remark} This result was obtained for quadratic finite fields in \cite{trans}. Although over the field $\mathbb F_{49}$ the Tsfasman--Vladut--Zink bound is better than the Gilbert--Varshamov bound on a non-empty interval, this interval does not include codes with $R=1/2$ (for $q=49$). Hence our proof works only for $q \geq 64$. 
\end{remark}

\vspace{3ex}

\noindent
Alp Bassa\\
Bo\u{g}azi\c{c}i University,\\
Faculty of Arts and Sciences,\\
Department of Mathematics,\\
34342 Bebek, \.{I}stanbul, Turkey\\
alp.bassa@boun.edu.tr\\

\noindent Henning Stichtenoth\\
Sabanc{\i} University, MDBF\\
{\rm 34956} Tuzla, \.Istanbul, Turkey\\
henning@sabanciuniv.edu\\
\end{document}